\documentclass{llncs}
\usepackage{amsmath}
\usepackage{amssymb}
\usepackage{hyperref}
\usepackage{algorithm}
\usepackage{algorithmic}
\usepackage{graphicx}

\newtheorem{thm}{Theorem}
\newtheorem{lem}[thm]{Lemma}

\newtheorem{prop}[thm]{Proposition}

\newtheorem{rem}[thm]{Remark}

\newtheorem{dfn}[thm]{Definition}

\newcommand{\dft}[1]{\textbf{\textit{#1}}}

\newcommand{\namedref}[2]{\hyperref[#2]{#1~\ref*{#2}}}
\newcommand{\sectionref}[1]{\namedref{Section}{#1}}

\newcommand{\subsectionref}[1]{\namedref{Subsection}{#1}}
\newcommand{\theoremref}[1]{\namedref{Theorem}{#1}}
\newcommand{\defref}[1]{\namedref{Definition}{#1}}

\newcommand{\lemmaref}[1]{\namedref{Lemma}{#1}}

\newcommand{\appref}[1]{\namedref{Appendix}{#1}}
\newcommand{\propref}[1]{\namedref{Proposition}{#1}}
\newcommand{\algref}[1]{\namedref{Algorithm}{#1}}

\newcommand{\equalityref}[1]{\hyperref[#1]{Equality~\eqref{#1}}}
\newcommand{\inequalityref}[1]{\hyperref[#1]{Inequality~\eqref{#1}}}
\newcommand{\remarkref}[1]{\namedref{Remark}{#1}}

\newcommand{\LCP}{\mathrm{LCP}}

\newcommand{\NLD}{\mathrm{NLD}}

\newcommand{\PLS}{\mathrm{PLS}}

\newcommand{\tPLS}{t\text{-}\mathrm{PLS}}

\newcommand{\acy}{\textsc{acyclic}}

\def\cA{{\cal A}}

\def\cC{{\cal C}}

\def\cF{{\cal F}}

\def\cP{{\cal P}}

\def\bV{\text{\bf v}}
\def\bP{\text{\bf p}}

\def\True{\textsc{true}}
\def\False{\textsc{false}}

\newcommand{\Cross}[3]{C(#1,#2,#3)}
\newcommand{\Neigh}[3]{N_#1(#2,#3)}

\newcommand{\abs}[1]{\left |#1\right |}
\newcommand{\paren}[1]{\left(#1\right)}
\newcommand{\set}[1]{\left\{#1\right\}}
\DeclareMathOperator{\dist}{dist}

\newcommand{\Add}{\textsc{Add}}
\newcommand{\Bounce}{\textsc{Bounce}}
\newcommand{\Count}{\textsc{Count}}
\newcommand{\Echo}{\textsc{Echo}}
\newcommand{\Increment}{\textsc{Increment}}
\newcommand{\RVerify}{\textsc{RVerify}}
\newcommand{\Send}{\textsc{Send}}
\newcommand{\Tcount}{T_{\mbox{count}}}

\newcommand{\Tleaf}{T_{\mbox{leaf}}}

\newcommand{\Tstart}{T_{\mbox{start}}}
\newcommand{\Tstop}{T_{\mbox{stop}}}
\newcommand{\Verify}{\textsc{Verify}}

\newcommand{\bhead}{\textsc{head}}
\newcommand{\bmid}{\textsc{mid}}
\newcommand{\btail}{\textsc{tail}}
\newcommand{\carry}{\textsc{carry}}
\newcommand{\continue}{\textsc{continue}}

\newcommand{\ctr}{\textsc{ctr}}
\newcommand{\halt}{\textsc{halt}}
\newcommand{\headcheck}{\textsc{head\_check}}
\newcommand{\iszero}{\textsc{is\_zero}}
\newcommand{\leaves}{\textsc{leaves}}
\newcommand{\msg}{\textsc{msg}}
\newcommand{\rec}{\textsc{rec}}
\newcommand{\send}{\textsc{snd}}

\newcommand{\val}{\textsc{val}}

\newcommand{\algorithmicassert}{\textbf{assert: }}
\newcommand{\ASSERT}{\STATE\algorithmicassert}

\newcommand{\dcount}{\textsc{dcount}}
\newcommand{\tcount}{\textsc{tcount}}

\newcommand{\firstPass}{\textsc{firstPass}}

\newcommand{\false}{\textsc{false}}

\DeclareMathOperator{\diam}{diam}

\title{Space-Time Tradeoffs for Distributed Verification}
\titlerunning{Space-Time Tradeoffs}

\author{%
  Rafail Ostrovsky\thanks{Research supported in part by NSF grant 1619348, DARPA, US-Israel BSF grant 2012366, OKAWA  Foundation  Research  Award,  IBM  Faculty  Research  Award,  Xerox  Faculty  Research Award, B. John Garrick Foundation Award, Teradata Research Award, and Lockheed-Martin Corporation Research Award. The views expressed are those of the authors and do not reflect position of the Department of Defense or the U.S. Government.}\inst{1}%
  \and Mor Perry\thanks{Partially supported by Apple Graduate Fellowship.}\inst{2}
  \and Will Rosenbaum\inst{2}
}

\institute{%
  Department of Computer Science and Department of Mathematics, University of California, Los Angeles, CA, USA
  \and School of Electrical Engineering, Tel Aviv University, Tel Aviv, Israel
}

\date{\today}

\begin{document}
	
  \maketitle

  \begin{abstract}
    Verifying that a network configuration satisfies a given boolean predicate is a fundamental problem in distributed computing. Many variations of this problem have been studied, for example, in the context of proof labeling schemes ($\PLS$), locally checkable proofs ($\LCP$), and non-deterministic local decision ($\NLD$). In all of these contexts, verification time is assumed to be constant. Korman, Kutten and Masuzawa~\cite{KKM} presented a proof-labeling scheme for MST, with poly-logarithmic verification time, and logarithmic memory at each vertex.

    In this paper we introduce the notion of a $\tPLS$, which allows the verification procedure to run for super-constant time. Our work analyzes the tradeoffs of $\tPLS$ between time, label size, message length, and computation space. We construct a universal $\tPLS$ and prove that it uses the same amount of total communication as a known one-round universal $\PLS$, and $t$ factor smaller labels. In addition, we provide a general technique to prove lower bounds for space-time tradeoffs of $\tPLS$. We use this technique to show an optimal tradeoff for testing that a network is  acyclic (cycle free). Our optimal $\tPLS$ for acyclicity uses label size and computation space $O((\log n)/t)$. We further describe a recursive $O(\log^* n)$ space verifier for acyclicity which does not assume previous knowledge of the run-time $t$.
  \end{abstract}



\section{Introduction}

A fundamental problem in distributed computing is to determine if 
a network configuration satisfies some predicate. In the distributed setting, a network configuration is represented by an underlying graph, where each vertex represents a processor, edges represent communication links between processors, and each vertex has a state. For example, the state of every vertex can be a color, and the predicate signifies that the coloring is proper, i.e., that every edge has its endpoints colored differently. Processors learn about the  network by exchanging messages along the edges. Some properties are local by nature and easy to verify, yet many natural problems---for example, testing if the network contains cycles---cannot be tested in less than diameter time, even if message size and local computational power are unbounded.

In order to cope with strong time lower bounds, Korman, Kutten, and Peleg introduced in~\cite{KKP} a computational model, called \emph{proof-labeling schemes} (PLS), where vertices are given auxiliary global information in the form of \emph{labels}. This auxiliary information may allow vertices to verify that a property is satisfied more efficiently than could be achieved without the aid of labels. Specifically, a PLS consists of two components, a \emph{prover} and a \emph{verifier}. The prover is an oracle which assigns labels to vertices. The verifier is a distributed algorithm which runs on the labeled configuration and outputs $\True$ or $\False$ at each vertex as a function of its state, its label, and the labels it receives. A PLS is \emph{complete} if for every legal configuration (satisfying the predicate), prover can assign labels such that all vertices output $\True$. The PLS is \emph{sound} if for every illegal configuration (which does not satisfy the predicate) for every labeling, some vertex outputs $\False$.

Schemes for verifying a predicate are useful in many applications. One such application is checking the output of a distributed algorithm~\cite{APV91,FRT13}. For example, if a procedure is meant to output a spanning-tree of the network, it may be useful to periodically verify that the output does indeed not contain cycles. 
If the original procedure which finds the spanning-tree can additionally produce labels, verification may be achieved substantially faster than diameter time required without the aid of labels. A simple procedure for checking the legality of the current state is very useful in the construction of self stabilizing algorithms~\cite{AO94,AKY,KKM,BFP14}. Other applications include estimating the complexity of logics required for distributed run-time verification~\cite{FRT13}, establishing a general distributed complexity theory~\cite{FKP13}, and proving lower bounds on the time required for distributed approximation~\cite{DH+12}. Local verification was recently applied in the design and analysis of software defined networks (SDN) in~\cite{Schmid2013}.

Distributed verification has been formalized in various models to suit its myriad applications. These models include proof-labeling schemes (PLS)~\cite{KKP}, locally checkable proofs (LCP)~\cite{GS11}, and non-deterministic local decision (NLD)~\cite{FKP13}. We refer the reader to~\cite{Feuilloley2016} for a detailed comparison of these models. All three of these models are local in the sense that verification requires a constant number of rounds, independent of the size of the graph. PLS differs from LCP and NLD in that verification in (traditional) PLS occurs in a single communication round, while the LCP and NLD models allow verification in a fixed constant number of rounds. While a fast procedure is certainly a desirable feature in verification algorithms, it may be the case that other computational resources---space or communication---must also be considered. For example, in the case of PLS, deterministically verifying a sub-graph is acyclic requires labels of size $\Omega(\log n)$ per vertex~\cite{KKP}. However, specifying a sub-graph only requires  $O(\Delta)$ space (the maximum degree of a vertex) per vertex. Thus, if we restrict attention to local verification algorithms, the space requirement to store labels may be unboundedly larger than the space required to specify the instance.

Korman, Kutten and Masuzawa~\cite{KKM} presented a PLS for minimum spanning-tree with poly-logarithmic verification time and logarithmic memory at each vertex. In the present work we also consider super-constant time verification and address tradeoffs between computational resources in distributed verification algorithms: label size, communication, computation space, and time. Specifically, we address the following questions: If verification algorithms are allowed to run in super-constant time, can labels be significantly shorter? What are the tradeoffs between label size and verification time? Can verification be achieved using (per processor) space which is linear in the label size? We focus on the acyclicity problem and prove that labels can indeed be shortened by a factor of $t$---the run-time of the algorithm---compared to constant-round verification. Moreover, computation space for each vertex can be made linear in the label size. Note that in this model it does not trivially hold that each message contains exactly one label, since in each round every vertex receives a (potentially different) label from each neighbor, and the scheme should specify the message to be sent in the following round. We show that in our schemes messages are small enough so that the total communication is the same as in one-round verification.

\subsection{Our Contributions}

In this paper we consider proof-labeling schemes with super-constant verification time, and analyze tradeoffs between time, label size, message size, and computation space. Many of the results presented here were announced without proof in~\cite{Baruch2016}. In~\subsectionref{sec:universal}, we describe a universal scheme which can verify any property $\cP$. Suppose $G_s$, with $n$ vertices, $m$ edges, and each state can be represented using $s$ bits. Then for every $t\in O(\diam(G_s))$, our scheme verifies $\cP$ in $t$ rounds using labels and messages of size  $O((ns+\min\{n^2,m\log n\})/t)$. For $t=1$ this is the known universal scheme~\cite{KKP,GS11,BFP}. When $t\in\Omega(n)$, we obtain labels and messages of size $O(s+\min\{n,(m/n)\log n\})$. Overall, labels are significantly smaller, and total communication is the same. \subsectionref{sec:lower-bound} proves a general lower bound technique for label size of $t$-round schemes.

In~\sectionref{sec:acy} we consider the problem determining if a graph is acyclic. Using the lower bound technique of~\subsectionref{sec:lower-bound}, we prove in~\subsectionref{sec:acy-lb} that labels of size $\Omega((\log n) /t)$ are required for the $\acy$ problem. \subsectionref{sec:certificate-upper-bound} shows that this lower bound is tight. Our scheme for $\acy$ additionally uses optimal space and messages of size $O((\log n)/t)$. In particular, by taking $t$ to be a sufficiently large constant, our upper bound (along with the $\Omega(\log n)$ lower bound for $\acy$ in~\cite{KKP}) implies separation between the PLS and LCP models for acyclicity (see~\cite{Feuilloley2016}). The verifier for $\acy$ assumes that vertices are given some truthful information about the round number, for example, by being told when (a multiple of) $t$ rounds have elapsed. We prove that such information is necessary for \emph{any} super-constant and sub-linear time distributed algorithm in~\appref{sec:impossibility}. In~\subsectionref{sec:recursive}, we describe a recursive scheme for $\acy$ which uses space $O(\log^* n)$ and constant communication per vertex per round. The recursive verifier runs in time $O(n)$ in the worst case, but there are always correct labels which will be accepted in time $O(\log \diam(G))$. We note that in order to break the logarithmic space barrier, our schemes in Subsections~\ref{sec:certificate-upper-bound} and~\ref{sec:recursive} crucially do not rely upon unique identifiers for the vertices. Conversely, the lower bounds of Subsections~\ref{sec:lower-bound} and~\ref{sec:acy-lb} hold for a stronger model where vertices have unique identifiers, and labels may depend on the unique identifiers.

\subsection{Related Work}

Distributed verification has been studied extensively. It was studied and used in the design of self stabilizing algorithms, first in~\cite{AKY}, where the notion of local detection was introduced, and recently in~\cite{KKM}, where a super-constant time verification scheme was presented. Both papers use verification in the design of a self stabilizing algorithm for constructing a minimum spanning-tree. Verification has also received attention of its own. For example,~\cite{KK07} presented tight bounds for minimum spanning-tree verification. In~\cite{KKP}, Korman, Kutten, and Peleg formalized the concept of local verification and introduced the notion of proof-labeling schemes. In their paper, verification is defined to use one communication round, and among other results they show a $\Theta(\log n)$ bound on the complexity (label size and communication) for $\acy$. Recently,~\cite{BFP} suggested using randomization in order to break the lower bounds of deterministic schemes, and among other results they show a $\Theta(\log \log n)$ bound on the communication complexity of acyclicity. In this paper, we show that if we use super-constant verification time, we can break the lower bound of space consumption (label size and computation space), while the total amount of communication is the same as in one  deterministic verification round. Proof-labeling schemes with constant, greater than one, verification time was studied in~\cite{GS11}, and with super-constant verification time was presented in~\cite{KKM}. In~\cite{Foerster2016}, the authors consider verification of acyclicity and related problems in various models for directed graphs.

The question of what properties can be verified using a constant verification time was studied in~\cite{FKP13}, and several complexity classes were presented, including LD---local decision---which includes all properties that can be decided using constant number of rounds and no additional information, and NLD---non-deterministic local decision---which includes all properties that can be decided in a constant number of rounds with additional information in the form of a certificate given to each vertex. While NLD and PLS are closely related, they differ in that NLD certificates are independent of vertex identifiers. Since PLS labels may depend on vertex identifiers, there is a PLS for every sequentially decidable property on ID based networks, while not all sequentially decidable properties are in NLD. Our lower bounds in Subsections~\ref{sec:lower-bound} and~\ref{sec:acy-lb} allow labels to depend on unique vertex identifiers, so our arguments give identical lower bounds for certificate sizes in the weaker NLD model. Nonetheless, the schemes for $\acy$ in Subsections~\ref{sec:certificate-upper-bound} and~\ref{sec:recursive} do not require unique identifiers.

Awerbuch and Ostrovsky describe a $\log^* n$-space distributed acyclicity verifier in~\cite{AO94}. Our scheme described in Section~\ref{sec:recursive} achieves the same space usage per node, but improves on the algorithm of~\cite{AO94} in several ways. The worst-case runtime of our acyclicity verifier is $O(n)$, whereas that in~\cite{AO94} requires time $O(n \log^2 n)$. Further, in our scheme there are always correct labels which are accepted in time $O(\log n)$. This runtime nearly matches the $\Omega((\log n) / \log^* n)$ time lower bound implied by Theorem~\ref{thm:acyclic-lower-bound}. We leave it as an open question if it is possible to verify $\acy$ using constant space and worst case runtime $O(\log n)$.


\section{Model and Definitions}
\label{sec:model}
\subsection{Computational Framework}
A \dft{graph configuration} $G_s$ consists of an underlying graph $G = (V,E)$, and a state assignment function  $\varphi : V \to S$, where $S$ is a state space. The state of a vertex includes all of its local information. It may include the vertex's identity (in an ID based configuration), the weight of its adjacent edges (in a weighted configuration), or the result of an algorithm executed on the graph, for example, its color according to a coloring algorithm.

In a proof-labeling scheme, an oracle assigns labels $\ell : V \to L$. Verification is performed by a distributed algorithm on the labeled configuration in synchronous rounds. In each round every vertex receives messages from all of its neighbors, performs local computation, and sends a message to all of its neighbors. At the beginning of each round, a vertex scans its messages in a streaming fashion, and the \dft{computational space} is the maximum space required by a vertex in its local computation. Each vertex may send different messages to different neighbors in a round. When a vertex halts, it outputs $\True$ or $\False$. If the vertex labels contain unique identifiers, then we require that an algorithm has the same output for all legal assignments of unique IDs.

\subsection{Proof-Labeling Schemes and \texorpdfstring{$\tPLS$}{t-PLS}}
\label{sec:pls}

We start with a short description of proof-labeling schemes (PLS) as introduced in \cite{KKP}. Given a family $\cF$ of configurations, and a boolean predicate $\cP$ over $\cF$, a PLS for $(\cF,\cP)$ is a mechanism for deciding $\cP(G_s)$ for every $G_s \in \cF$. A PLS consists of two components: a \dft{prover} $\bP$, and a \dft{verifier} $\bV$. The prover is an oracle which, given any configuration $G_s\in\cF$, assigns a bit string $\ell(v)$ to every vertex $v$, called the \dft{label} of $v$. The verifier is a distributed algorithm running concurrently at every vertex. 
The verifier $\bV$ at each vertex outputs a boolean. If the outputs are $\True$ at all vertices, $\bV$ is said to \dft{accept} the configuration, and otherwise (i.e., $\bV$ outputs $\False$ in at least one vertex) $\bV$ is said to \dft{reject} the configuration. For correctness, a proof-labeling scheme $(\bP, \bV)$ for $(\cF,\cP)$ must be (1) \dft{complete} and (2) \dft{sound}. Formally, for every $G_s\in\cF$, we say $(\bP, \bV)$ is

\begin{enumerate}
\item \dft{complete} if $\cP(G_s) = \True$ then, using the labels assigned by $\bP$, the
  verifier $\bV$ accepts $G_s$, and
\item \dft{sound} if $\cP(G_s) = \False$ then, for every label assignment, the
  verifier $\bV$ rejects $G_s$.
\end{enumerate}

The \dft{verification complexity} of a proof-labeling scheme $(\bP, \bV)$, according to \cite{KKP}, is the maximal label size---the maximal length of a label assigned by the prover $\bP$ on a legal configuration (satisfying $\cP$). A PLS is defined to use one verification round, in which neighbors exchange labels. In this case, label size and message size are the same. 

In this paper we consider proof-labeling schemes with more than one verification round, in particular it can use super-constant time, and hence we define the \dft{message size} of the scheme  $(\bP, \bV)$ to be the largest message a vertex sends during the execution of $\bV$ on a legal configuration with the labels assigned by $\bP$. We denote a proof-labeling scheme with $t$-round verification by $t$-$\PLS$.


\section{General Space-Time Tradeoff Results}
\label{sec:space-time}

If there exists a $\PLS$ for $(\cF,\cP)$ with label size $\kappa$ (and hence, message size $\kappa$), then there exists a $t$-$\PLS$ for $(\cF,\cP)$ with label size $\kappa$ and message size $\kappa/t$. Indeed, vertices can communicate their $\kappa$-bit label in $t$ different \dft{shares} of size $\kappa/t$. In this section we give general results for label size reduction, along with message size, in a $t$-$\PLS$. The idea is to take a $1$-$\PLS$, and break it into smaller shares where vertices are assigned only a single share of the original label. We refer to this technique as \dft{label sharing}. In particular, we present a universal scheme and provide a tool for obtaining lower bounds.

\subsection{Universal \texorpdfstring{$\tPLS$}{t-PLS}}
\label{sec:universal}

A \dft{universal scheme} is a scheme that verifies every sequentially decidable property. In this subsection we assume that every vertex has an identifier, and identifiers in the same configuration are pairwise distinct. We give an upper bound on the label and message size of a universal scheme that uses $t$ communication rounds.

\begin{thm}
  \label{thm:universal}
  Let $\cF$ be a family of configurations with states set $S$ and diameter at least $D$, let $\cP$ be a boolean predicate over $\cF$ and suppose that every state in $S$ can be represented using $s$ bits. For every $t\in \Omega(D)$ there exists a $t$-$\PLS$ for $(\cF,\cP)$ with label and message size $O((ns+\min\{n^2,m\log n\})/t)$ where $n$ is the number of vertices, and $m$ is the number of edges in the graph.
\end{thm}

In the proof of this theorem we use a known universal $\PLS$~\cite{KKP,GS11,BFP}. Labels consist of the entire representation of the graph configuration. Nodes then verify that they have the same representation, and that it is consistent with its local view. Finally, they verify individually that the label represents a legal configuration. Since every configuration can be represented using $O({ns+\min\{n^2,m\log n\}})$ bits---by listing the state of each vertex and an adjacency matrix or an edge list---this is the label (and message) size of this scheme. 

The idea of the universal $\tPLS$ is to disperse the configuration representation into shares such that each vertex can collect the purported graph configuration from its $t$-neighborhood. 

\begin{proof}[of \theoremref{thm:universal}]
  Let $\cF$ be a family as described in the statement, let $\cP$ be a boolean predicate over $\cF$ and $G_s=(V,E,\varphi : V \to S)\in \cF$. We first describe the scheme. Consider some fixed vertex $v\in V$. For every vertex $u\in V$, let $\dist(u,v)=d$ and define $j \equiv d \mod (t/4)$. Denote $R = (ns+\min\{n^2,m\log n\})$. The \dft{universal label} of $u$, denoted by $c(u)$, consists of:
  \begin{itemize}
  \item a \dft{v-indication} $d_0(u) \in \set{0, 1}$ indicating if $u=v$,
  \item a \dft{first in block} $f(u) \in \set{0, 1}$ indicating if $j=0$,
  \item an \dft{orientation label} $a(u) \in \set{0, 1, 2}$ encodes $(d\ mod \ 3)$, and
  \item a \dft{share of representation} $r(u) \in \set{0, 1}^{(4R) / t}$ which encodes the $j$-th part (out of $t/4$ parts, of length $\frac{4R}{t}$ each) of $G_s$'s representation.
  \end{itemize}
  
  In the first round, each vertex sends its label to all of its neighbors. In the first $t/2$ rounds we use the orientation indicated by the orientation label of each neighbor for an efficient pipelining of labels in two directions. The message of every vertex in each of the first $t/2$ rounds is composed of two parts, one for pipelining of labels towards $v$ and the other for pipelining of labels away from $v$. For every vertex $u \in V$, let $Y_{(-1)}$ be all neighbors $y$ of $u$ with $a(y)\equiv a(u)-1\mod 3$, and let $Y_{(+1)}$ be all neighbors $y$ of $u$ with $a(y)\equiv a(u)+1\mod 3$. The pipelining towards $v$ is done by receiving labels only from $Y_{(+1)}$ and sending labels only to $Y_{(-1)}$. Let $L^i_{(+1)}$ be the set of labels $u$ received in round $i$ from all its $Y_{(+1)}$ neighbors. The vertex $u$ verifies that all non empty labels in $L^i_{(+1)}$ are equal, and sends this label to $Y_{(-1)}$. The pipelining away from $v$ is done similarly, with the roles of $Y_{(-1)}$ and $Y_{(+1)}$ reversed. The distinguished vertex $v$ verifies that it has only $Y_{(+1)}$ neighbors, and in each round all non empty labels in $L^i_{(+1)}$ are equal, and sends this label to all its neighbors. Every vertex $u\neq v$ verifies that during the first $t/2$ rounds it has received from $Y_{(-1)}$ two labels (in two different rounds) with `first in block' indication, $f = 1$. If the first had also `$v$-indication' then $u$ concatenates all `shares of representation' of these labels, in order, excluding the last. Otherwise (the first had no `$v$-indication'), $u$ concatenates all `shares of representation' of these labels, in reverse order, excluding the first. The distinguished vertex $v$ verifies that it has `$v$-indication', `first in block' indication, and `orientation label' $0$, and concatenates the $t/4$ first `shares of representation' it sees, in order (including $r(v)$). Every vertex $u\in V$ considers its concatenation, denoted by $g(u)$, as a representation of a configuration, and verifies that it is consistent with its local view. In the last $t/2$ rounds $u$ verifies that for every neighbor $w$ it holds that $g(w)=g(u)$, by sending $g(u)$ in $t/2$ disjoint shares. Finally, if all verifications succeed, the output of $u$ is whether the configuration represented by $g(u)$ satisfies $\cP$.
  
  The label size is $O(R/t)$. In the first $t/2$ rounds, every message contains exactly two labels, and hence message size is also $O(R/t)$. For every $u$, by definition, $g(u)$ is the concatenation of at most $t/2$ `shares of representation' ($t/2$ rounds, and at most one `share of representation' is concatenated in each round). Therefore, in the last $t/2$ rounds every message size is not more than the size of one `share of representation', which is also $O(R/t)$. So, the label and message size requirements hold.
  
  We now prove the correctness of the scheme.  If all vertices output $\True$, by the last part of the scheme we know that they all have the same representation, and that it is consistent with their local view. Therefore, it must be the case where all vertices hold the correct representation of $G_s$. Since all vertices output $\True$, by construction of the scheme, $\cP(G_s)=\True$. If $\cP(G_s)=\True$ and labels are assigned according to the scheme, we have the following. Denote by $c_j$ the label of a vertex with distance $j$ from $v$. Let $u\in V$ be a vertex and let $\dist(u,v)=d$. In round $i$, by construction of the scheme, $u$ receives from $Y_{(-1)}$ (and $v$ from $Y_{(+1)}$) the label $c_{|d-i|}$. If $d< t/4$, by construction, the first label $u$ receives with `first in block' indication (after less than $t/4$ rounds) is $c_0$. Afterwards it receives $c_1, c_2,\ldots, c_{t/4-1}$ and $c_{t/4}$ which is the second with `first in block' indication. If $d\ge t/4$, the first label $u$ receives with `first in block' indication (after less than $t/4$ rounds) is not $c_0$, and hence has no `$v$-indication'. By construction, it must be $c_{Z}$, where $Z=t/4\cdot k$ for some natural number $k>0$. Afterwards it receives $c_{Z-1}, c_{Z-2},\ldots, c_{Z-t/4+1}$ and $c_{Z-t/4}$ which is the second with `first in block' indication.  It is easy to see that in both cases $u$ constructs the correct representation of $G_s$. Therefore, the equality and local view verifications succeed, and since $\cP(G_s)=\True$, all vertices output $\True$.  
\end{proof}

 
\subsection{Lower Bound Tool}
\label{sec:lower-bound}

We start with some definitions. Although we consider only networks represented by undirected graphs, we will define an orientation on an edge to indicate a specific ordering of its endpoints. We denote by $H(e)$ the head of a directed edge $e$, and by $T(e)$ the tail of $e$.
 
\begin{dfn}[Edge Crossing]
  Let $G=(V,E)$ be a graph, and $e_1,e_2\in E$ be two directed edges. The \emph{edge crossing} of $e_1$ and $e_2$ in $G$, denoted by $\Cross{e_1}{e_2}{G}$, is the graph obtained from $G$ by replacing $e_1$ and $e_2$, by the edges $(T(e_1),H(e_2))$ and $(T(e_2),H(e_1))$.
\end{dfn}

Edge crossings were used many times before, and were formalized as a tool for proving lower bounds of verification complexity in~\cite{BFP}. We now show how to use edge crossing in order to prove lower bounds for label size of $\tPLS$.

\begin{dfn}[Edge $k$-neighborhood]
  Let  $G=(V,E)$ be a graph, and $e=(u,v)\in E$. The \dft{$k$-neighborhood} of $e$ in $G$, denoted by $\Neigh{k}{e}{G}$, is the subgraph $(V',E')$ of $G$ satisfying
  \begin{enumerate}
  \item $w\in V'$ if and only if $w\in V$ and $\min(\dist(w,u),\dist(w,v))\le k$, and
  \item $e'\in E'$ if and only if $e'\in E\cap (V'\times V')$.
  \end{enumerate}
\end{dfn}

\begin{prop}
  \label{prop:tool}
  Let $(\bP,\bV)$ be a deterministic $t$-$\PLS$ for $(\cF,\cP)$ with label size $\abs{\ell}$. Suppose that there is a configuration $G_s\in\cF$ which satisfies $\cP$ and contains $r$ directed edges $e_1,\dots,e_r$, whose $t$-neighborhoods $\Neigh{t}{e_1}{G_s},\dots,\Neigh{t}{e_r}{G_s}$ are pairwise disjoint, contain $q$ vertices each, and there exist $r$ state preserving isomorphisms $\{\sigma_i:V(\Neigh{t}{e_1}{G_s})\to V(\Neigh{t}{e_i}{G_s}), i=1,\dots,r\}$ such that $\sigma_i(H(e_1))=H(e_i)$ and $\sigma_i(T(e_1))=T(e_i)$. If $\abs{\ell}<(\log r)/q$, then there exist $i, j$ with $1\le i< j\le r$ such that every connected component of $\Cross{e_i}{e_j}{G_s}$ is accepted by $(\bP,\bV)$.
\end{prop}


\begin{proof}
  Let $(\bP,\bV)$ and $G_s$ be as described above, and assume that $\abs{\ell}<(\log r)/q$. Consider a collection $\{\sigma_i:V(\Neigh{t}{e_1}{G_s})\to V(\Neigh{t}{e_i}{G_s}), i=1,\dots,r\}$ of $r$ state preserving isomorphisms,
  such that $\sigma_i(H(e_1))=H(e_i)$ and $\sigma_i(T(e_1))=T(e_i)$. Order the vertices of $\Neigh{t}{e_1}{G_s}$ arbitrarily. For every $i$, consider the concatenation of labels given by $\bP$ to the vertices of $\Neigh{t}{e_i}{G_s}$, in the order induced by the ordering of $\Neigh{t}{e_1}{G_s}$ and $\sigma_i$. Denote this concatenated string $L_i$. By label size assumption, it holds that $|L_i|<\log r$ for every $i$, and thus there are less than $r$ different options for $L_i$. Therefore, by the pigeonhole principle, there are $i\neq j$ such that $L_i=L_j$. Denote $\Cross{e_i}{e_j}{G_s}$ by $G'_s$, and consider the labels provided by $\bP$ to $G_s$. For every vertex $v\notin \Neigh{t}{e_i}{G_s}\cup \Neigh{t}{e_j}{G_s}$, its $t$-neighborhood is the same in $G_s$ and in $G'_s$. $\Neigh{t}{e_i}{G_s}$ and $\Neigh{t}{e_j}{G_s}$ are disjoint, isomorphic, and have the same states and labels according to some isomorphism which maps $H(e_i)$ to $H(e_j)$ and $T(e_i)$ to $T(e_j)$. Thus, for every vertex $v\in \Neigh{t}{e_i}{G_s}\cup \Neigh{t}{e_j}{G_s}$, its $t$-neighborhood in $G_s$ is the same as in $G'_s$. Since the output of the verifier $\bV$ at each vertex in $G_s$ is only a function of the states and labels at its $t$-neighborhood, if the output of $\bV$ in $G_s$ is $\True$ at all vertices, then the output of $\bV$ in every connected component of $G'_s$ must be $\True$, and the proposition follows.
\end{proof}

The following theorem, which is a consequence of \propref{prop:tool}, is the tool we use to prove lower bounds of label size in a $t$-$\PLS$. 

{\sloppy
\begin{thm}
  \label{thm:tool}
  Let $\cF$ be a family of configurations, and let $\cP$ be a boolean predicate over $\cF$. Suppose that there is a configuration $G_s\in\cF$ which satisfies
  \begin{enumerate}
  \item $\cP(G_s)=\True$,
  \item $G_s$ contains $r$ directed edges $e_1,\dots,e_r$, whose $t$-neighborhoods $\Neigh{t}{e_1}{G_s},\dots,\Neigh{t}{e_r}{G_s}$ are pairwise disjoint, contain $q$ vertices each, and there exist $r$ state preserving isomorphisms $\{\sigma_i:V(\Neigh{t}{e_1}{G_s})\to V(\Neigh{t}{e_i}{G_s}), i=1,\dots,r\}$ such that $\sigma_i(H(e_1))=H(e_i)$ and $\sigma_i(T(e_1))=T(e_i)$, and
  \item for every $i\ne j$, there exists a connected component $H_s$ of $\Cross{e_i}{e_j}{G_s}$ such that $\cP(H_s)=\False$.
  \end{enumerate}
  Then the label size of any $t$-$\PLS$ for $(\cF,\cP)$ is $\Omega((\log r)/q)$.
\end{thm}
}


\section{Acyclicity}
\label{sec:acy}
In this section we focus on the acyclicity property, and give tight $t$-$\PLS$ lower and upper bounds. The lower bounds of \subsectionref{sec:acy-lb} hold in the computational model where vertices have unique identifiers, and the labels are allowed to depend on the ID of a vertex. The upper bounds presented in Subsections~\ref{sec:certificate-upper-bound} and~\ref{sec:recursive} still apply in a weaker computational model where vertices do not have unique IDs. 

\begin{dfn}[Acyclicity]
Let $\cF$ be the family of all connected graphs. Given a graph configuration $G_s\in \cF$,  $\acy(G_s)=\True$ if and only if the underlying graph $G$ is cycle free.
\end{dfn}

\subsection{Lower Bound for \texorpdfstring{$\acy$}{acyclic}}
\label{sec:acy-lb}

\begin{thm}
  \label{thm:acyclic-lower-bound}
  Every scheme which verifies $\acy$ in $t$ communication rounds requires labels of size $\Omega\left((\log n) / t\right)$.
\end{thm}

\begin{proof}
  We will show a configuration as described in \theoremref{thm:tool}, with $r=\Omega\left(n/t\right)$ and $q=O(t)$, to derive the stated lower bound on label size of any scheme that verifies $\acy$. Let $G_s$ be the $n$-vertex path $v_0-v_1- \dots -v_{n-1}$ where all states are the empty string. Obviously $\acy(G_s)=\True$. Let $r=\left\lfloor n/(2t+2)\right\rfloor-1$, and consider the set $\{e_i=(v_{(2t+2)i},v_{(2t+2)i+1})\mid 1 \le i \le r\}$ of $r$ directed edges. Each $\Neigh{t}{e_i}{G_s}$ contains exactly $2t+2$ vertices, and thus $q=2t+2$. Every pair of $t$-neighborhoods $\Neigh{t}{e_i}{G_s}$ and $\Neigh{t}{e_j}{G_s}$, for $i\not= j$, is disjoint since the distance between $e_i$ and $e_j$ is at least $2t+1$. 
  For every $i < j$, $\Cross{e_i}{e_j}{G_s}$ contains exactly two connected components. One of them is the cycle $H_s=v_{qi+1}-v_{qi+2}-\dots - v_{qj} -v_{qi+1}$ where all its edges are marked. By definition, $\cP(H_s)=\False$. Hence, the conditions of \theoremref{thm:tool} are satisfied, and the lower bound follows.
\end{proof}

\subsection{Upper Bound for \texorpdfstring{$\acy$}{acyclic}}
\label{sec:certificate-upper-bound}

In this section, we describe a $\tPLS$ for $\acy$ which matches the lower bound presented in \theoremref{thm:acyclic-lower-bound}. 

\begin{thm}
  \label{thm:certificate-upper-bound}
  Suppose $G = (V, E)$ is a graph with diameter $D$. For every $t \leq \min\set{\log n, D}$, there exists an $O(t)$-$\PLS$ for $\acy$ with label and messages of size $O((\log n) / t)$. Further, the verifier $\bV$ uses space of size $O((\log n) / t)$.
\end{thm}

\begin{rem}
  In this subsection, we assume that each vertex has access to some means of deciding (correctly) when $t$ communication rounds have elapsed. This can be achieved either by allowing each vertex a $\log t$ bit counter, or by giving each vertex access to an oracle which alarms when (an integer multiple of) $t$ rounds have elapsed. We discuss the necessity of this assumption in \subsectionref{sec:recursive}, and prove that such information is necessary for any distributed algorithm with super-constant and sub-linear run-time in \appref{sec:impossibility}.
\end{rem}

The following scheme can be used to verify that the graph contains no cycles using labels of size $O(\log n)$ in a single round. The label of a vertex $v$ consists of an integer $d(v)$ which encodes the distance from $v$ to a root vertex (which has $d(v) = 0$). Vertices verify the correctness of the labels in a single communication round. If $v$ satisfies $d(v) = 0$ (i.e., $v$ is a root), then it accepts the label if all of its neighbors $w$ satisfy $d(w) = 1$. If $v$ satisfies $d(v) \not= 0$ then $v$ verifies that $v$ has exactly one neighbor $u$ with $d(u) = d(v) - 1$ while all other neighbors $w$ satisfy $d(w) = b(v) + 1$. This scheme is used, for example, in~\cite{APV91,IL94,AO94}. The correctness of the scheme is a consequence of the following definition and lemma. 

\begin{dfn}
  \label{dfn:cyclic-label}
  Suppose $G = (V, E)$ is a graph and $L = \set{0, 1, \ldots, s-1}$ with $s \geq 3$. We call function $\ell : V \to L$ an \dft{$s$-cyclic labeling} of $G$ if for every $v \in V$, $v$ has at most one neighbor $P(v)$---the \dft{parent} of $v$---such that $\ell(P(v)) \equiv \ell(v) - 1 \mod s$, while the $v$'s other neighbors $w$ satisfy $\ell(w) \equiv \ell(v) + 1 \mod s$. 
\end{dfn}

\begin{rem}
  \label{rem:cyclic-label}
  An $s$-cyclic labeling induces an orientation on $G$ where an edge $(u, v)$ is oriented such that $u = P(v)$. That is, each edge is oriented away from the parent.
\end{rem}

\begin{lem}
  \label{lem:cyclic-label}
  Suppose $G = (V, E)$ is a connected graph and $\ell$ an $s$-cyclic labeling. Then either $G$ is acyclic or $G$ contains a unique cycle of length $k$, where $s$ divides $k$. Further, if $G$ contains a cycle, $C$, then $C$ is an oriented cycle in the orientation induced by $\ell$, and all oriented paths in $G$ are oriented away from vertices in $C$. 
\end{lem}

\begin{proof}
  Suppose $C = (v_0, v_1, \ldots, v_{k - 1})$ is a cycle in $G$. In the orientation described in \remarkref{rem:cyclic-label}, every vertex has in-degree at most $1$. Let $\deg_{in}(v_i)$ denote the in-degree of $v_i$ in $C$ and similarly $\deg_{out}(v_i)$ is $v_i$'s out-degree in $C$. Then $\deg_{in}(v_i) - \deg_{out}(v_i) \leq 0$ for all $v_i$. However, we must have $\sum_{i} \deg_{in}(v_i) - \deg_{out}(v_i) = 0$, implying that in fact $\deg_{in} v_i= \deg_{out}(v_i) = 1$ for all $i$. Thus, $C$ is an oriented cycle. As a consequence, for all $i$, either $\ell(v_i) \equiv \ell(v_{i+1}) + 1 \mod s$ or $\ell(v_i) \equiv \ell(v_{i+1}) - 1 \mod s$. In the former case, we have $\ell(v_{k - 1}) - \ell(v_{0}) \equiv k \equiv 0 \mod s$, implying that $s$ divides $k$. In the latter case, $\ell(v_{k - 1}) - \ell(v_{0}) \equiv - k \equiv 0 \mod s$, and the desired result holds.

  Since every vertex $v_i \in C$ has in-degree $1$ in $C$, all edges that leave $C$ must be oriented away from vertices in $C$. Similarly, any path $w_0, w_1, \ldots, w_j$ with $w_0 \in C$ and $w_i \notin C$ for $i \geq 1$  must be oriented away from $C$. Thus no such path may lead to another cycle $C'$, nor could another cycle $C'$ share a path with $C$. Thus since $G$ is connected $C$ must the unique cycle.
\end{proof}

To achieve labels of length $O((\log n) / t)$ for $\acy$, we simulate the ``distance-to-root'' scheme described above. The idea is to break the $O(\log n)$-bit labels indicating the distance to the root into shares of size $O((\log n)/t)$. Unlike the universal scheme described in \subsectionref{sec:universal}, vertices do not reconstruct the $(\log n)$-bit distance-to-root labels directly, but check the labeling is correct distributively. Thus the verifier $\bV$ only uses space linear in the label size.

\begin{figure}[t]
  \includegraphics[width=\textwidth]{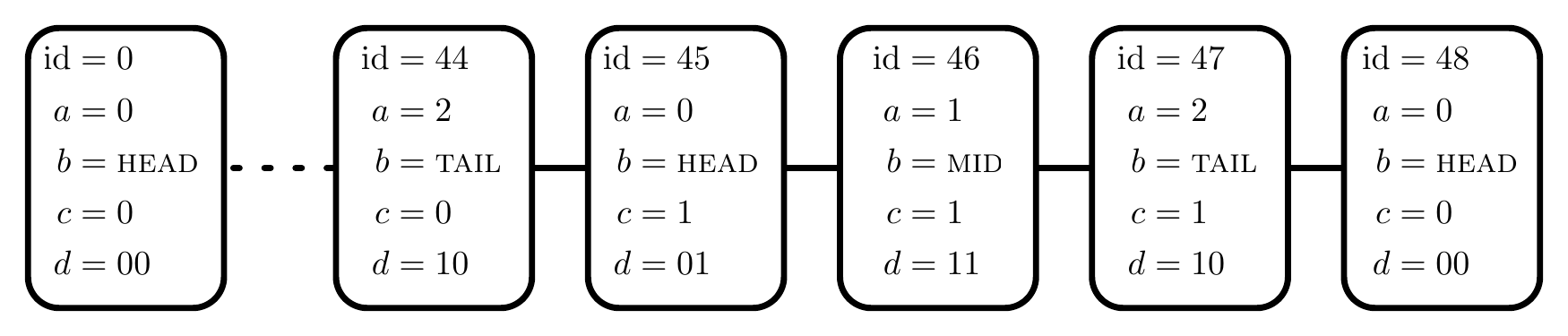}
  \caption{Acyclicity labels for a graph consisting of a path rooted at its left endpoint. We have given the nodes identifiers $0, 1, \ldots$ from left to right, although the labeling need not include the id of the vertices. For this configuration, the orientation labels $a(v)$ simply count the distance from $v$ to the root (with id $0$) modulo $3$. The nodes with ids $45, 46$, and $47$ form a single block, whose head (45) and tail (47) are indictated by the corresponding block labels. The color of this block is $1$ because it is the 15th block from the root ($45 / 3 = 15$), and $15 \equiv 1 \mod 2$. Finally, the concatination of the distance labels in this block is $d(47)d(46)d(45) = 101101$, which encodes the distance of the block's head to the root ($45$) in binary.}
  \label{fig:block-labeling}
\end{figure}

Formally, for a vertex $v$, an \dft{acyclicity label} consists of:
\begin{itemize}
\item an \dft{orientation label} $a(v) \in \set{0, 1, 2}$ which defines an orientation on edges away from the root of the tree, 
\item a \dft{block label} $b(v) \in \set{\bhead, \bmid, \btail}$ which indicates $v$'s position within a block,
\item a \dft{block color} $c(v) \in \set{0,1}$, and
\item a \dft{distance label} $d(v) \in \set{0, 1}^{(\log n) / t}$ which encodes a share of a distance to the root.
\end{itemize}
See Figure~\ref{fig:block-labeling} for an example of correctly formed labels. It is clear that an acyclicity label can be recorded in $O((\log n) / t)$ bits. The semantics of acyclicity labels are described below.
\begin{description}
\item[Correct orientation labels] The orientation labels $a(v)$ are correct if every $v \in V$ has at most one neighbor $P(v)$---the \dft{parent} of $v$---such that $a(P(v)) \equiv a(v) - 1 \mod 3$. The remaining neighbors $w$ of $v$---$v$'s \dft{children}---satisfy $a(w) \equiv a(v) + 1 \mod 3$. If $P(v) = \varnothing$, we call $v$ a \dft{root}. Correct orientation labels induce an orientation on $G$ where the oriented edges $(v, w)$ satisfy $a(w) \equiv a(v) + 1 \mod 3$. Thus, edges are oriented away from roots (if any).
\item[Correct block labels] Block labels must be assigned in the following manner
  \begin{enumerate}
  \item $b(v) = \bhead$ if and only if either $P(v) = \varnothing$ or $b(P(v)) = \btail$
  \item $b(v) = \btail$ if and only if there exists an oriented path of length $t$, $v_0, v_1, \ldots, v_{t-1} = v$ such that $b(v_0) = \bhead$. We refer to such a path as a \dft{block}.
  \item In all other cases, $b(v) = \bmid$.
  \item For every $v$, there exists an oriented path $w_0, w_1, \ldots, w_{k-1} = v$ of length $k < t$ such that $b(w_0) = \bhead$.
  \end{enumerate}
\end{description}

\begin{dfn}
  Let $B = (v_0, v_1, \ldots, v_{t-1})$ be a block. We define the \dft{value} of $B$, denoted $D(B)$, to be the integer whose binary expansion is the concatenation $d(v_{t-1}) d(v_{t-2}) \cdots d(v_{0})$. That is, $v_0$ holds the least significant bits of $D(B)$, while $v_{t-1}$ holds the most significant bits. If $B' = (w_0, w_1, \ldots, w_{t-1})$ is another block, we say that $B$ is the \dft{parent} of $B'$ and $B'$ is a \dft{child} of $B$ if $P(w_0) = v_{t-1}$. If there exists $i$ such that $v_i = w_i$, we say that $B$ and $B'$ \dft{overlap}.
\end{dfn}

\begin{description}
\item[Correct block coloring] The block coloring $c$ is correct if
  \begin{enumerate}
  \item for every block $B$ and $v, w \in B$ we have $c(v) = c(w)$, and
  \item for every blocks $B, B'$ such that $B$ is the parent of $B'$, and $v \in B$, $w \in B'$, we have $c(v) \neq c(w)$.
  \end{enumerate}
\item[Correct distance labels] The distance labels $d$ are correct if
  \begin{enumerate}
  \item for every block, $B = (v_0, v_1, \ldots, v_{t-1})$, $D(B) = 0$ if and only if $v_0$ is a root, and
  \item for every pair of blocks $B$ and $B'$ with $B$ the parent of $B'$, we have $D(B') = D(B) + t$.
  \end{enumerate}
\end{description}

\begin{dfn}[Correct acyclicity labeling]
  \label{dfn:correct-acy-cert}
  Suppose $\ell$ is a family of acyclicity labels for a graph $G = (V, E)$. We say that the family $\ell$ is \dft{correct} if $a$, $b$, $c$, and $d$ are correct orientation labels, correct block labels, correct block colorings, and correct distance labels as described above. 
\end{dfn}

\begin{rem}
  \label{rem:overlap}
  If blocks $B = (v_0, \ldots, v_{t-1})$ and $B' = (w_0, \ldots, w_{t-1})$ overlap, then we must have $w_0 = v_0$ and $D(B) = D(B')$. The first equality holds because each vertex $v_i$ has at most one parent, so if $w_i = v_i$ we must have $w_j = v_j$ for $0 \leq j \leq i$. The second equation holds because either $B$ and $B'$ contain a root, in which case $D(B) = D(B') = 0$, or there is a $B''$ which is the parent of both $B$ and $B'$. In the latter case, $D(B) = D(B'') + t = D(B')$.
\end{rem}

\begin{prop}
  \label{prop:acyclicity-certificate}
  Let $G = (V, E)$ be a graph. Then $G$ is acyclic if and only if it admits a correct labeling $\ell$.
\end{prop}

\begin{proof}
  If $G$ is acyclic, then we can form labels $\ell$ in the following way. Choose an arbitrary vertex $u$ to be the root. For all $v$ define $d'(v) = \dist(v, u)$ (the length of the unique path from $v$ to $u$), and take $a(v) = d'(v) \mod 3$. Define $b(v)$ by $b(v) = \bhead$ if $d'(v) \equiv 0 \mod t$, $b(v) = \btail$ if $d'(v) \equiv -1 \mod t$, and $d(v) = \bmid$ otherwise. Finally, assign distance labels $d(v)$ in such a way that in each block $B$ with first element $v_0$, $D(B) = d'(v_0)$. It is easy to verify that these labels $\ell$ constructed in this way will satisfy all the provisions of \defref{dfn:correct-acy-cert}.

  Conversely, suppose $G$ admits a correct family of acyclicity labels. Suppose towards a contradiction that $C = (w_0, w_1, \ldots, w_{k - 1})$ is a cycle. Since the orientation labels $a(v)$ are correct (hence form a $3$-cyclic labeling), $C$ must be an oriented cycle (as in the proof of \lemmaref{lem:cyclic-label}). The final provision in the correctness of $b$ and the fact that each vertex $w_i$ has a unique parent guarantee some $w_i$ must have $b(w_i) = \bhead$. Without loss of generality, assume that $b(w_0) = \bhead$, and let $B_0$ be the block containing $w_0$ and contained in $C$. Inductively define blocks $B_1, B_2, \ldots \subseteq C$ such that $B_{i+1}$ is a child of $B_i$. By the pigeonhole principle, we must have $B_i = B_j$ for some $i < j$. However, the correctness of the distance labels implies that $D(B_i) < D(B_{i+1}) < \cdots < D(B_j) = D(B_i)$, a contradiction.
\end{proof}

In order to prove \theoremref{thm:certificate-upper-bound}, by \propref{prop:acyclicity-certificate}, it suffices to show there is a verifier $\bV$ for acyclicity labels which runs in time $O(t)$ using messages and memory of size $O((\log n) / t)$. Verification of the correctness of the orientation labels $a$, block coloring $c$, and conditions 1 and 3 in the correctness of the block labels $b$ can be accomplished in a single communication round with constant communication. Thus, we must verify conditions 2 and 4 in the correctness of the block labels as well as the correctness of distance labels. 

After the initial sharing of labels with neighbors in the first round, the verification algorithm $\Verify(v, a, b, c, d)$ continues as follows (see \algref{alg:verify} for pseudo-code). For $t-1$ steps, each vertex relays the message from its parent to all of its children. At the end of $t$ rounds, each vertex verifies that at some point, it received a message from a head vertex. If a vertex $v$ received a message from a root vertex, it verifies that $d(v) = 0$. Otherwise, let $b(w)$, $c(w)$, and $d(w)$ be labels received by $v$ in the $t$-th round. Then $v$ checks that $b(w) = b(v)$, $c(w) \neq c(v)$. The block heads increment the distance labels $d(w)$ $t$ times, sending carry bits (if any) to their children. When children receive carry bits, they increment their $d(w)$'s accordingly, sending further carry bits to their children. After this incrementation procedure, vertex $v$ verifies that the incremented $d(w)$'s satisfy $d(v) = d(w)$.

\begin{algorithm}
  \caption{$\Verify(v, a, b, c, d)$: Verifies correctness of acyclicity labels.}
  \label{alg:verify}
  \begin{multicols}{2}
  \begin{algorithmic}[1]
    \STATE send $a(v)$, $b(v)$, and $c(v)$ to all neighbors
    \STATE verify correctness of $a$ and $c$, and conditions 1 and 3 in correctness of $b$
    \STATE $\headcheck \leftarrow \FALSE$
    \IF{$b(v) = \btail$}
    \STATE $\iszero \leftarrow \TRUE$
    \ENDIF
    \FOR{i = 1 \TO t-1}
    \STATE $M \leftarrow (b(w), c(w), d(w))$ or $\varnothing$ received from $P(v)$ 
    \IF{$b(w) = \bhead$}
    \STATE $\headcheck \leftarrow \TRUE$
    \ENDIF
    \IF{$b(v) = \btail$}
    \IF{$d(w) \neq 0$}
    \STATE $\iszero \leftarrow \false$
    \ENDIF
    \IF{$i = t - 1$}
    \ASSERT $b(w) = \bhead$
    \ENDIF
    \ENDIF
    \STATE send $M$ to all children \COMMENT{if $v$ is a leaf, ignore}
    \ENDFOR
    \IF{$M = \varnothing$}
    \ASSERT $d(v) = 0$ \COMMENT{head of $v$'s block is root}
    \ELSE
    \FOR{$i = 1$ \TO $t$}
    \STATE $\Increment(d(w), \abs{d(w)}, 1)$
    \ENDFOR
    \ASSERT $b(w) = b(v)$
    \ASSERT $c(w) \neq c(v)$
    \ASSERT $d(w) = d(v)$
    \IF{$b(v) = \btail$}
    \ASSERT $\iszero = \FALSE$
    \ENDIF
    \ENDIF
    \ASSERT $\headcheck = \TRUE$
  \end{algorithmic}
  \end{multicols}
\end{algorithm}

\begin{lem}
  \label{lem:verify}  
  Let $\ell$ be a family of acyclicity labels on a graph $G = (V, E)$. Then $\ell$ is correct if and only if every vertex $v$ accepts in \algref{alg:verify}.
\end{lem}
\begin{proof}
  By induction, each vertex receives the message from its (unique) $i$-th ancestor in the $i$-th communication round. Therefore, every tail accepts at lines 16--18 if and only if every tail is at (oriented) distance $t-1$ from a head. Similarly, every vertex $v$ is at (oriented) distance $i_v < t$ from a head if and only if it accepts at line 35 (see lines 9--11). Thus, the block labels are correct if and only if every vertex accepts at lines 2, 17, and 35.

  Note that $b(w) = \varnothing$ if and only if the head of the block containing $v$ is a root. Thus, every vertex accepts at line~23 if and only if all blocks $B$ containing a root satisfy $D(B) = 0$. Conversely, if $B$ does not contain a root, then by the assertion at line~32 (and the check at lines 13--15), then $D(B) \neq 0$. Thus the checks at lines 23 and 32 are satisfied if and only if condition 1 in the correctness of distance labels is satisfied. 

  Suppose block $B = (w_0, \ldots, w_{t-1})$ is the parent of $B' = (v_0, \ldots, v_{t-1})$, then the distance label received  by each $v_i$ is $d(w_i)$. Thus, after incrementing the labels $d(w_0) d(w_1) \cdots d(w_{t-1})$ $t$ times, the incremented labels will have value $D(B) + t$. Therefore, all vertices in $B'$ accept at line 31 if and only if $D(B') = D(B) + t$, if and only if condition 2 of correct distance labels is satisfied. 
\end{proof}

\begin{proof}[of \theoremref{thm:certificate-upper-bound}]
  \lemmaref{lem:verify} implies that the $\Verify$ routine (\algref{alg:verify}) is a correct verifier for acyclicity labels. Thus we must only argue that $\Verify$ achieves the claimed time, space, and communication bounds. In each communication round, each vertex broadcasts a single label (in line 20) or a single bit (in $\Increment$) to its neighbors. Thus, the communication in each round is $O((\log D) / t)$ per edge. In each iteration of the algorithm, each vertex stores at most a constant number of labels, hence the memory usage is $O((\log D) / t)$ as well. Finally, the overall run-time is $3 t$. The label sending procedure in lines 7--21 is accomplished in $t$ rounds, while the incrementation procedure in lines 25--7 requires at most $2 t$ rounds: $t$ rounds where the head vertices increment, and another $t$ to propagate carries. In particular, the run-time is $O(t)$.
\end{proof}


\subsection{Recursive Acyclicity Checking}
\label{sec:recursive}

The scheme described in \subsectionref{sec:certificate-upper-bound} gives asymptotically optimal label size for $t \leq \log n$. Further, the communication per round and local memory usage is linear in the label size. However, the scheme above crucially requires each vertex to be given a truthful representation of the parameter $t$. In fact, for $\omega(1) \leq t \leq o(n)$, it is necessary for the vertices to be given some truthful information about $t$ (see \appref{sec:impossibility}). In this subsection, we describe a verifier for $\acy$ that only assumes that the space provided to each processor is $O(\log^* n)$. The tradeoff is that our algorithm runs in time which may be linear in $n$ in the worst case.

\begin{thm}
  \label{thm:recursive-verify}
  There exists a $O(n)$-$\PLS$ for $\acy$ which uses labels and space of size $O(\log^*n)$. In each round, the communication per-edge is $O(1)$.
\end{thm}

\begin{rem}
  While verification time in \theoremref{thm:recursive-verify} is $O(n)$ in the worst case, the actual time depends on the labels given to the vertices. In particular, for every acyclic graph $G$ there exists a correct labeling which will be accepted in time $O(\log D)$. Thus there is a tradeoff between the time of the algorithm and the amount of truthful information about $t$ given to the vertices.
\end{rem}

The idea of the algorithm is to simulate the verifier $\Verify$ (\algref{alg:verify}) without the benefit of truthful information about $t$. As before, the labels designate blocks of length $t$. Within each block, the vertices store shares of the distance of that block to the root, where in this case, the shares consist of a single bit. Since $t$ (the length of the block) is not known to the vertices in advance, they must first compute $t$. However, storing $t$ requires $\log t$ bits, so the computed value of $t$ is stored in shares in sub-blocks of length $\log t$.
In order to verify the correctness of the sub-blocks, the vertices must count to $\log t$ using $\log \log t$ bits of memory. This value is again stored in shares in sub-sub-blocks of length $\log \log t$. This process of recursively verifying the lengths of blocks continues until the block length is constant. Thus $\log^* n$ levels of recursion suffice.

Formally, in our recursive scheme, \dft{recursive acyclicity labels} closely resemble those in \subsectionref{sec:certificate-upper-bound}. For each vertex $v$ and each level $i = 1, 2, \ldots, k = \log^* n$, we have an associated block label $b_i(v)$ and block color $c_i(v)$. We refer to the labels associated to each $i$ as a \dft{level}, denoted $L_i$. The top level $L_1$ additionally contains orientation labels, $a(v)$ and distance labels $d(v)$ for each vertex. Each level $i$ has an associated length, denoted by $t_i$. We emphasize that the $t_i$ are not initially known to the vertices at the beginning of an execution. The semantics and correctness of the block labels $b_i$ and block colors $c_i$ are precisely the same as those described in \subsectionref{sec:certificate-upper-bound}, where blocks at level $i$ have length $t_i$. As before, the distance labels $d(v)$ encode (a share of) the purported distance of the $L_1$ block containing $v$ to the root.

\begin{dfn}
  \label{dfn:recursive-acy-label}
  Suppose $\ell$ is a family of recursive acyclicity labels for a graph $G = (V, E)$. We say that a family $\ell$ of recursive acyclicity labels is \dft{correct} if the $L_1$ labels are correct as in \defref{dfn:correct-acy-cert}, and for $i \geq 2$ the block labels in $b_i$ and block colors $c_i$ are correct as in \defref{dfn:correct-acy-cert} with $t_i = \lfloor \log t_{i-1} \rfloor$.
\end{dfn}

\begin{rem}
  For simplicity of presentation, we assume that for all $i \geq 2$ that $t_i$ divides $t_{i-1}$. Thus, each block in $L_{i-1}$ contains an integral number of sub-blocks. The general case can be obtained by allowing ``overlap'' of the last sub-block of $B$ in level $i$ with the first sub-block of $B'$ in $i$ where $B$ is the parent block of $B'$.
\end{rem}


Analogously to \propref{prop:acyclicity-certificate}, we obtain the following result.

\begin{prop}
  \label{prop:recursive-acy-label}
  Let $G = (V, E)$ be a graph. Then $G$ is acyclic if and only if it admits a correct family $\cC$ of recursive acyclicity labels.
\end{prop}

It is clear that recursive acyclicity labels are of length $O(\log^* n)$. Indeed, each of the labels in the $\log^* n$ recursive levels has length $O(1)$.

\begin{lem}
  \label{lem:recursive-acy-verify}
  Let $G = (V, E)$ be a graph, and $\cC$ a family of recursive acyclicity labels on $G$. Suppose that for some $i$, the labels in $L_{i+1}$ are correct. Then there exists a verifier $\bV_i$ for the labels in $L_i$ with run-time $O(2^{t_{i+1}})$, constant communication per round, and constant space.
\end{lem}


\begin{algorithm}
  \caption{$\RVerify(i, L_i)$}
  \label{alg:r-verify-label}
  \begin{multicols}{2}
  \begin{algorithmic}[1]
    \STATE verify $a$ is correct
    
    \STATE verify properties 1 and 3 of correctness of $b_i$ correctness of $c_i$

    \IF{$i = \log^* n$}
    \STATE verify correctness of $b_i$ and $c_i$
    \RETURN
    \ENDIF

    \STATE $\tcount_{i+1} \leftarrow 0$
    \STATE $\Count(\tcount_{i+1}, 1, i)$
    \STATE $\Send(\tcount_{i+1}, \rec, i+1)$
    \ASSERT $\rec_{i+1} = \tcount_{i+1}$

    \IF{$i = 1$}
    \STATE $\Add(d(v), \tcount_2, \dcount, 1)$
    \STATE $\Send(\dcount, \dcount, 1)$
    \ASSERT $\dcount = d(v)$
    \ENDIF
    
  \end{algorithmic}
  \end{multicols}
\end{algorithm}

We describe a verifier $\RVerify$ (\algref{alg:r-verify-label}) for $L_i$ assuming $L_{i+1}$ is correct. 
Suppose $B$ is a block in level $i$, and $B_1, B_2, \ldots, B_s$ its sub-blocks for $s = t_i / t_{i+1}$, with $B_j$ the parent of $B_{j+1}$. By assumption, the block labels for the $B_j$ are correct. The head $v_0$ of $B$ verifies that it is also the head of $B_1$, and sends a token $\Tcount$ to all of its children. The vertices in $B$ bounce $\Tcount$ to the tail, which then bounces $\Tcount$ back up to $v_0$. Meanwhile, the vertices of each $B_j$ hold shares of a counter $\tcount_j$, which computes $t_i$ by incrementing itself until $\Tcount$ returns to the head. If the counter $\tcount_j$ ever exceeds $2^{t_{i+1}}$ (i.e., if the bit held by the tail of $B_j$ is ever incremented twice), then the vertices in $B_j$ will halt and reject the label. It is clear that this step of the verification will always halt in time $O(2^{t_{i+1}})$. After counting, the blocks in $L_{i+1}$ verify that they agree on $\tcount_j$. Further, tails of $B_j$ verify that their share of $\tcount$ is $1$, implying that $2^{t_{i-1} - 1} < t_i \leq 2^{t_{i-1}}$.

There is a slight complication in the verification algorithm described above that arises when a block $B$ terminates prematurely in a leaf (a vertex of degree 1) which is not a tail. In correct block labels, if $v_0$ is the head of overlapping \dft{complete} blocks (i.e., all have tails at distance $t_i$ from the head) then $v_0$ should receive $\Tcount$ from all of its children at the same time, $2 t_i$. However, if some block containing $v_0$ is \dft{incomplete} (terminates prematurely with a leaf) then $v_0$ may receive messages from its children in different rounds. To avoid this problem, leaves which are not labeled $\btail$ respond with a token $\Tleaf$ to their parent upon receiving $\Tcount$. The parent then knows not to expect a $\Tcount$ from this child. Similarly, if an internal vertex receives $\Tleaf$ from all of its children (perhaps in different rounds), it sends $\Tleaf$ to its parent. Then vertices check that they receive $\Tcount$ from all children at the same time, except those which have sent $\Tleaf$ if a previous round.

Finally, if $i = 1$, the vertices must additionally verify the correctness of the distance labels $d(v)$. 
Suppose $B = (v_0, \ldots, v_{t-1})$ and $B' = (w_0, \ldots, w_{t-1})$ are blocks with $B$ the parent of $B'$. The tail $v_{t-1}$ sends $b(v_{t-1})$, $c(v_{t-1})$, and $d(v_{t-1})$ to its children, and sends the token $\Tstart$ to its parent, $v_{t-2}$. The vertices continue to echo any messages received from their parents to their children, and if a vertex $v$ receives $\Tstart$ from its children, it additionally sends $b(v)$, $c(v)$, and $d(v)$ to its children. When $w_{t-1}$ (the tail of $B'$) receives $d(v_{t-1})$, it saves this value and sends $\Tstop$ to its parent. When a vertex $w$ receives $\Tstop$, it saves the value $d(v)$ in the message it received from its parent such that $c(v) \neq d(v)$, and echos $\Tstop$ to its parent. After $2 t$ rounds, the procedure terminates, and every $w_i$ holds $d(v_i)$. In a further $3 t$ rounds, $B'$ distributively increments the $d(v_i)$, and verify that the incremented $d(v_i)$ are equal to $d(w_i)$, thus ensuring the distance labels are correct.

\begin{proof}[of \lemmaref{lem:recursive-acy-verify}]
  We prove that $\RVerify(i, L_i)$ (\algref{alg:r-verify-label}) is a verifier for $L_i$ whenever $L_{i+1}$ is a correct. As in the proof of \lemmaref{lem:verify}, we focus on verifying properties 2 and 4 in the correctness of $b_i$. Properties 1 and 3 of the correctness of $b_i$, as well as the correctness of $c_i$ can be trivially verified in a single communication round with constant communication. Let $v_0$ be a root in $L_i$. By induction, every vertex at distance $\tau$ from $v_0$ receives $\Tcount$ at time $\tau$. Thus, property 4 of the correctness of $b_i$ is satisfied if and only if no vertex fails in a call to $\Count(\tcount_{i+1}, 1, i)$, which occurs if and only if each $2^{t_{i+1} - 1} < \tcount_{i+1} \leq 2^{t_{i+1}}$ (line 20 of $\Count$ ensures the first inequality, while the check in lines 11--13 of $\Increment$ ensure the second inequality). Property 2 in the correctness of $b_i$ holds if and only if all vertices accept the assertion at line 10 of $\RVerify(i, L_i)$.

  The proof that $d$ is correct when $i = 1$ if and only if no vertex rejects in lines 11--15 in $\RVerify(i, L_i)$ is analogous to the argument in \lemmaref{lem:verify}. Finally, it is clear that the per-round communication is constant, as is the space requirement (assuming that only levels $L_i$ and $L_{i+1}$ are stored). As for the run-time, notice that $\Count(\ctr, m, i)$ always terminates in time at most $2^{m t_{i+1}}$ by the verification at lines 11--13 of $\Increment$. Further, if no vertex fails during the call to count $\Count$, then $\Add$ and $\Send$ will similarly halt after $2^{t_{i+1}} \leq t_i$ rounds.
\end{proof}

\begin{proof}[of \theoremref{thm:recursive-verify}]
  By \propref{prop:recursive-acy-label}, it suffices to prove the existence of a verifier $\bV$ of recursive acyclicity labels with the claimed communication, space, and time. We induct on $k - i$ (where $k = \log^* n$) that the correctness of $L_{i}$ can be verified in the desired run-time, using constant communication and space. When $i = k$, the correctness of labels is a local property (independent of the size of the network). Thus, each vertex $v$ can verify the correctness of $L_{k}$ by analyzing the state of $L_k$ labels in $N(v, O(1))$, which can be accomplished in constant time, space, and communication. Now suppose the correctness of $L_{i+1}$ can be verified in time $O(t)$ using constant communication and space. By \lemmaref{lem:recursive-acy-verify}, $\RVerify(i, L_i)$ (\algref{alg:r-verify-label}) is a verifier for $L_i$. Further, $\RVerify(i, L_i)$ runs in time $O(t_i) \leq O(\log(t_1))$, uses constant communication, and space. \theoremref{thm:recursive-verify} the follows by running $\RVerify(k, L_k)$, followed by $\RVerify(k-1, L_{k-1})$ and so on, up to $\RVerify(1, L_1)$. The run-time is $O(t_k + t_{k-1} + \cdots + t_1) \leq O(t_1)$.
\end{proof}

\begin{rem}
  We can modify the recursive scheme described here to use only finitely many levels of recursion, but with the tradeoff of using more memory per-vertex. In particular, if only the labels of $L_1$ are given, but each vertex has access to a counter with $\log t$ bits of memory, we recover precisely the scheme of \subsectionref{sec:certificate-upper-bound} in the case where $t = \Omega(\log n)$. If we give labels in $L_1$ and $L_2$, and each vertex has a counter with $\log \log t$ bits of memory, then the scheme will still be correct. However, we get a greater degradation of run-time due to round-off errors in $\log \log t$. Specifically, if we have $m - 1 < \log \log t \leq m$, then we obtain
  \[
  2^{2^{m-1}} < t \leq \paren{2^{2^{m-1}}}^2.
  \]
  Thus, even if $\log \log t$ is given truthfully as the size of the counter, the run-time of $\RVerify$ may be quadratic in $t$ if the $L_1$ labels are improperly formed. Finally, given labels $L_1$, $L_2$, and $L_3$, and a counters of size $\log^{(3)} t$, the run-time may vary exponentially from $\log n$. Thus, our worst-case run-time is already only $O(n)$. The fully recursive scheme thus achieves the same worst-case run-time with $\log^* n$ memory per vertex.
\end{rem}





\bibliographystyle{abbrv}
\bibliography{local-verification}

\begin{thebibliography}{10}

\bibitem{AKY}
Y.~Afek, S.~Kutten, and M.~Yung.
\newblock The local detection paradigm and its application to
  self-stabilization.
\newblock {\em Theor. Comput. Sci.}, 186(1-2):199--229, 1997.

\bibitem{AO94}
B.~Awerbuch and R.~Ostrovsky.
\newblock Memory-efficient and self-stabilizing network reset (extended
  abstract).
\newblock In {\em Proceedings of the Thirteenth Annual ACM Symposium on
  Principles of Distributed Computing}, PODC '94, pages 254--263, New York, NY,
  USA, 1994. ACM.

\bibitem{APV91}
B.~Awerbuch, B.~Patt-Shamir, and G.~Varghese.
\newblock Self-stabilization by local checking and correction.
\newblock In {\em 32nd Symposium on Foundations of Computer Science (FOCS)},
  pages 268--277. IEEE, 1991.

\bibitem{BFP}
M.~Baruch, P.~Fraigniaud, and B.~Patt{-}Shamir.
\newblock Randomized proof-labeling schemes.
\newblock In {\em Proceedings of the 2015 {ACM} Symposium on Principles of
  Distributed Computing, {PODC}}, pages 315--324, 2015.

\bibitem{Baruch2016}
M.~Baruch, R.~Ostrovsky, and W.~Rosenbaum.
\newblock Brief announcement: Space-time tradeoffs for distributed
  verification.
\newblock In {\em Proceedings of the 2016 ACM Symposium on Principles of
  Distributed Computing}, PODC '16, pages 357--359, New York, NY, USA, 2016.
  ACM.

\bibitem{BFP14}
L.~Blin, P.~Fraigniaud, and B.~Patt-Shamir.
\newblock On proof-labeling schemes versus silent self-stabilizing algorithms.
\newblock In {\em 16th Int. Symp. on Stabilization, Safety, and Security of
  Distributed Systems (SSS)}, LNCS, pages 18--32. Springer, 2014.

\bibitem{DH+12}
A.~{Das Sarma}, S.~Holzer, L.~Kor, A.~Korman, D.~Nanongkai, G.~Pandurangan,
  D.~Peleg, and R.~Wattenhofer.
\newblock Distributed verification and hardness of distributed approximation.
\newblock {\em SIAM J. Comput.}, 41(5):1235--1265, 2012.

\bibitem{Feuilloley2016}
L.~Feuilloley and P.~Fraigniaud.
\newblock Survey of distributed decision.
\newblock {\em Bulletin of the {EATCS}}, 119, 2016.

\bibitem{Foerster2016}
K.-T. Foerster, T.~Luedi, J.~Seidel, and R.~Wattenhofer.
\newblock Local checkability, no strings attached.
\newblock In {\em Proceedings of the 17th International Conference on
  Distributed Computing and Networking}, ICDCN '16, pages 21:1--21:10, New
  York, NY, USA, 2016. ACM.

\bibitem{FKP13}
P.~Fraigniaud, A.~Korman, and D.~Peleg.
\newblock Towards a complexity theory for local distributed computing.
\newblock {\em J. ACM}, 60(5):35, 2013.

\bibitem{FRT13}
P.~Fraigniaud, S.~Rajsbaum, and C.~Travers.
\newblock Locality and checkability in wait-free computing.
\newblock {\em Distributed Computing}, 26(4):223--242, 2013.

\bibitem{GS11}
M.~G{\"o}{\"o}s and J.~Suomela.
\newblock Locally checkable proofs.
\newblock In {\em 30th ACM Symp. on Principles of Distributed Computing
  (PODC)}, pages 159--168, 2011.

\bibitem{IL94}
G.~Itkis and L.~Levin.
\newblock Fast and lean self-stabilizing asynchronous protocols.
\newblock In {\em Proceedings of the 35th Annual Symposium on Foundations of
  Computer Science}, SFCS '94, pages 226--239, Washington, DC, USA, 1994. IEEE
  Computer Society.

\bibitem{KK07}
A.~Korman and S.~Kutten.
\newblock Distributed verification of minimum spanning trees.
\newblock {\em Distributed Computing}, 20:253--266, 2007.

\bibitem{KKM}
A.~Korman, S.~Kutten, and T.~Masuzawa.
\newblock Fast and compact self stabilizing verification, computation, and
  fault detection of an {MST}.
\newblock In {\em 30th Annual {ACM} Symposium on Principles of Distributed
  Computing (PODC)}, pages 311--320, 2011.

\bibitem{KKP}
A.~Korman, S.~Kutten, and D.~Peleg.
\newblock Proof labeling schemes.
\newblock {\em Distributed Computing}, 22(4):215--233, 2010.

\bibitem{Schmid2013}
S.~Schmid and J.~Suomela.
\newblock Exploiting locality in distributed sdn control.
\newblock In {\em Proceedings of the Second ACM SIGCOMM Workshop on Hot Topics
  in Software Defined Networking}, HotSDN '13, pages 121--126, New York, NY,
  USA, 2013. ACM.

\end{thebibliography}


\appendix







\section{Super-constant and sub-linear algorithms}
\label{sec:impossibility}

In this section, we show that any algorithm $\cA$ which has run-time which is $\omega(1)$ and $o(n)$ for all inputs must have access to some truthful global information about $G$ or $t$. Suppose $G = (V, E)$ is a graph, $S$ a (possibly infinite) set of states, and $\varphi : V \to S$ an assignment of initial states. In the $\tau$-th step of computation, each vertex $v$ learns the state of its neighbors up to distance $\tau$, and must decide to halt or continue. Thus, we can view an algorithm as a function $f$ on from labeled graphs to the set $\set{\halt, \continue}$. On the $\tau$-th step, the vertex $v$ computes $f(N(v, \tau))$ either halts or continues based on the value of $f$. We say that $\cA$ halts in time $t$ on input $(G, \varphi)$ if every vertex halts in time $\tau \leq t$ and some vertex $v$ halts precisely at time $t$. We say that $\cA$ has run-time $t$ on $G$ if for all initial inputs $\varphi$ for $G$, the run-time of $(G, \varphi)$ is at most $t$, and there exists some initial input for which the run-time is $t$. We denote the run-time of $\cA$ on $G$ by $t(G)$.

\begin{prop}
  \label{prop:impossibility}
  Let $\cC = \set{C_3, C_4, \ldots}$ denote the family of cycle graphs. Suppose the sequence of run-times $t(C_3), t(C_4), \ldots$ is unbounded. Then $t(C_n) = \Omega(n)$. 
\end{prop}

\begin{proof}
  Since $t(C_n)$ is unbounded, define $n_k$ to be the smallest value of $n$ for which $t(C_n) \geq k$. Suppose $\varphi : C_{n_k} \to S$ gives initial states for which the run-time is at least $k$, and in particular, that the vertex $v$ does not halt after $k-1$ rounds. Let $v_{-k+1}, \ldots, v_{-1}, v, v_1, \ldots, v_k$ denote $v$'s $k-1$ neighborhood.

  Now consider $C = C_{2k}$. Fix $w \in C$ and let $w_{-k+1}, \ldots, w_{k-1}$ denote $w$'s $k-1$ neighborhood. Let $\psi : C \to S$ be an initial assignment which satisfies $\psi(w_i) = \psi(v_i)$ for all $i = -k+1, \ldots, k-1$. Thus, $N(v, k-1)$ and $N(w, k-1)$ are isomorphic. In particular, this implies that $\cA$ will not halt at $w$ in fewer than $k$ rounds. Thus, $t(C_{2k}) \geq k$. Therefore, for all $n$, we have $t(C_n) \geq n/2$, which gives the desired result.
\end{proof}

\clearpage

\section{Pseudocode for Subroutines}
{\footnotesize
\begin{algorithm}
  \caption{$\Count(\ctr, m, i)$: Computes the length of a block with shares of the count $\ctr \in \set{0, 1}^m$.}
  \label{alg:count}
  \begin{multicols}{2}
  \begin{algorithmic}[1]
    \STATE $\firstPass \leftarrow \FALSE$
    \IF{$b_i(v) = \bhead$}
    \STATE send $\Tcount^i$ to children
    \STATE $\firstPass \leftarrow \TRUE$
    \ENDIF
    
    \REPEAT
    \STATE $\Increment(\ctr, m, 1, i+1)$
    \STATE $\Bounce(\Tcount^i, i)$
    \UNTIL{$\firstPass = \TRUE$}

    \REPEAT
    \STATE $\Increment(\ctr, m, 1/2, i+1)$
    \STATE $\Bounce(\Tcount^i, i)$
    \STATE $\Echo(\Tstop^i)$
    \IF{$b_i(v) = \mathbf{head}$ and received $\Tcount^i$ from children}
    \STATE send $\Tstop^i$ to children
    \RETURN
    \ENDIF
    \UNTIL{receive $\Tstop^i$ from $P(v)$}
    \IF{$b_i(v) = \btail$}
    \ASSERT $\ctr > 2^{m-1} - 1$
    \ENDIF
  \end{algorithmic}
  \end{multicols}
\end{algorithm}

\begin{algorithm}
  \caption{$\Increment(\ctr, m, \val, i)$: Increments the counter $\ctr \in \set{0, 1}^m$ by $\val$, sending carry bits to children. Rejects if new count exceeds capacity of block.}
  \label{alg:increment}
  \begin{multicols}{2}
  \begin{algorithmic}[1]
    \IF{$b(v) = \bhead$}
    \STATE $\carry \leftarrow (\ctr + \val) / 2^m$
    \STATE $\ctr \leftarrow (\ctr + \val) \mod 2^m$
    \ELSIF{receive value $\val'$ from parent}
    \STATE $\carry \leftarrow (\ctr + \val') / 2^m$
    \STATE $\ctr \leftarrow \ctr + \val'$
    \ENDIF
    \IF{$\carry \neq 0$}
    \IF{$b(v) = \bhead$ or $\bmid$}
    \STATE send $\carry$ to all children
    \ELSE
    \STATE reject
    \ENDIF
    \ENDIF
  \end{algorithmic}
  \end{multicols}
\end{algorithm}

\begin{algorithm}
  \caption{$\Bounce(T, i)$: bounces a token $T$ from head to tail, and back to the head. Fails if $T$ is received from different children at different times.}
  \label{alg:bounce}
  \begin{multicols}{2}
  \begin{algorithmic}[1]
    \STATE $\leaves \leftarrow \varnothing$
    \IF{receive $T$ from $P(v)$}
    \IF{$b_i(v) = \btail$}
    \STATE{send $T$ to $P(v)$}
    \ELSIF{$v$ has no children}
    \STATE send $\Tleaf^i$ to $P(v)$
    \ELSE
    \STATE send $T$ to all children
    \ENDIF
    \ENDIF
    \IF{receive $\Tleaf^i$ from child $w$}
    \STATE $\leaves \leftarrow \leaves \cup \set{w}$
    \ENDIF
    \IF{$\leaves$ contains all children}
    \STATE send $\Tleaf^i$ to $P(v)$
    \ELSIF{receive $T$ from all children $\notin \leaves$}
    \STATE send $T$ to $P(v)$
    \ELSE
    \STATE fail
    \ENDIF
  \end{algorithmic}
  \end{multicols}
\end{algorithm}

\begin{algorithm}
  \caption{$\Echo(T, i)$: sends a token $T$ from head to tail}
  \label{alg:echo}
  \begin{multicols}{2}
  \begin{algorithmic}[1]
    \IF{receive $T$ from $P(v)$}
    \IF{$b_i(v) = \textrm{mid}$}
    \STATE send $T$ to all children
    \ELSIF{$b_i(v) = \textrm{head}$}
    \STATE fail
    \ENDIF
    \ELSIF{receive $T$ from any child}
    \STATE fail
    \ENDIF
  \end{algorithmic}
  \end{multicols}
\end{algorithm}

\begin{algorithm}
  \caption{$\Send(\msg, \rec, i)$: sends messages $\msg$ stored in block $B$ to $B'$, $B$'s child, stores message as $\rec$}
  \label{alg:send}
  \begin{multicols}{2}
  \begin{algorithmic}[1]
    \IF{$b_i(v) = \btail$}
    \STATE send $\set{\msg, c_i(v)}$ to all children
    \STATE send $\Tstart^i$ to $P(v)$
    \REPEAT
    \STATE $\send \leftarrow$ message from $P(v)$
    \IF{$\send$ contains $\set{\msg(w), c_i(w)}$ with $c_i(w) \neq c_i(v)$}
    \STATE $\rec \leftarrow \msg(w)$
    \STATE send $\Tstop^i$ to $P(v)$
    \STATE $\send \leftarrow \send\setminus\set{\msg(w), c_i(w)}$
    \ENDIF
    \STATE send $\send$ to all children
    \UNTIL{receive $\Tstop^i$ from $P(v)$}
    \ELSE
    \REPEAT
    \STATE $\send \leftarrow$ message from $P(v)$
    \IF{receive $\Tstop^i$ from all children}
    \STATE $\rec \leftarrow \msg(w)$ where $c_i(w) \neq c_i(v)$
    \STATE $\send \leftarrow \send\setminus\set{\msg(w), c_i(w)}$
    \ENDIF
    \STATE send $\send$ to all children
    \STATE $\Bounce(\Tstop^i)$
    \IF{received $\Tstop^i$ from all children and $b_i(v) = \bhead$}
    \STATE send $\Tstop^i$ to all children
    \RETURN
    \ENDIF
    \UNTIL{receive $\Tstop^i$ from $P(v)$}
    \ENDIF
  \end{algorithmic}
  \end{multicols}
\end{algorithm}

\begin{algorithm}
  \caption{$\Add(\ctr_1, \ctr_2, \ctr_3, m)$: adds $\ctr_1$ and $\ctr_2$ and stores result as $\ctr_3$; all counters are in $\set{0, 1}^m$}
  \label{alg:add}
  \begin{multicols}{2}
  \begin{algorithmic}[1]
    \IF{$b(v) = \bhead$}
    \STATE $\carry \leftarrow (\ctr_1 + \ctr_2) / 2^m$
    \STATE $\ctr_3 \leftarrow (\ctr_1 + \ctr_2) \mod 2^m$
    \STATE send $\carry$ to all children
    \ELSIF{receive $\val$ from $P(v)$}
    \STATE $\carry \leftarrow (\ctr_1 + \ctr_2 + \val) / 2^m$
    \STATE $\ctr_3 = (\ctr_1 + \ctr_2 + \val) \mod 2^m$
    \STATE send $\carry$ to all children
    \ENDIF
  \end{algorithmic}
  \end{multicols}
\end{algorithm}
}

\end{document}